\documentclass[letterpaper, 10 pt, conference]{ieeeconf}

\usepackage{amsmath, amssymb, amsfonts, mathtools}  

\usepackage{theorem}
\usepackage{epsfig}
\usepackage[mathscr]{eucal}
\usepackage{cases}
\usepackage{bm,bbm,dsfont}
\usepackage{xcolor,ifpdf}

\usepackage{graphics,epstopdf} 
\usepackage{graphicx}           

\let\labelindent\relax
\usepackage{enumerate}
\usepackage[shortlabels]{enumitem}
\setlist[description]{leftmargin=\parindent}
\setlist[itemize]{leftmargin=*}

\usepackage{float}
\usepackage{subfigure}

\usepackage{tikz}
\usepackage{algorithmic}
\newtheorem{algorithm}{Algorithm}
\renewenvironment{proof}{\noindent\textit{Proof:}\ }{\hfill$\blacksquare$\par\medskip}
\usepackage{algorithm}
\usetikzlibrary{arrows,graphs} 
\usetikzlibrary{calc}

\newcommand*\mcapinn[2]{\vcenter{\hbox{$\mathsurround=0pt
			\ifx\displaystyle#1\textstyle\else#1\fi\bigcap$}}}

\newcommand*\mcupinn[2]{\vcenter{\hbox{$\mathsurround=0pt
			\ifx\displaystyle#1\textstyle\else#1\fi\bigcup$}}}

\DeclareFontFamily{OT1}{pzc}{}
\DeclareFontShape{OT1}{pzc}{m}{it}{<-> s * [1.200] pzcmi7t}{}
\DeclareMathAlphabet{\mathpzc}{OT1}{pzc}{m}{it}

\newtheorem{thm}{Theorem}

\newtheorem{lem}{Lemma}

\newtheorem{assum}{Assumption}

{\theorembodyfont{\upshape}

}

\newcommand{\beq}{\begin{equation}}
\newcommand{\eeq}{\end{equation}}
\newcommand{\beqa}{\begin{eqnarray}}
\newcommand{\eeqa}{\end{eqnarray}}
\newcommand{\beqan}{\begin{eqnarray*}}
	\newcommand{\eeqan}{\end{eqnarray*}}

\newcommand{\bite}{\begin{itemize}}
	\newcommand{\eite}{\end{itemize}}
\newcommand{\benu}{\begin{enumerate}}
	\newcommand{\eenu}{\end{enumerate}}

\theoremstyle{theorem}
\theorembodyfont{\color{violet}\itshape}

\IEEEoverridecommandlockouts                              
\title{\LARGE \bf Hybrid Sequential Feedback Optimization for Wind Farm Power Maximization}

\author{Shijie Huang and Sergio Grammatico \thanks{S. Huang and S. Grammatico are with the Delft Center of Systems and Control, TU Delft, The Netherlands. E-mail address: \texttt{S.Huang-5@tudelft.nl; s.grammatico@tudelft.nl}.}\thanks{This work was partially supported by the NWO under research project Online Optimization for Offshore Wind Farms.
}}
\date{}

\begin{document}
	
	\maketitle
	\thispagestyle{empty}
	\pagestyle{empty}
	
	\begin{abstract}
		This paper considers feedback optimization for optimal steady-state operation of nonlinear discrete-time systems when the steady-state input-output map and its sensitivity are expensive to compute. We propose a hybrid extension of  sequential feedback optimization (SFO) that augments the model-based SFO gradient with correction terms through a convex combination with summable diminishing weights. Two variants are studied: one based on recursive least-squares (RLS) sensitivity estimation, and another on extremum seeking control (ESC) gradient estimation. Under contractivity and smoothness assumptions, we show that both hybrid schemes preserve the convergence of SFO to a neighborhood of the optimal steady state. The proposed methods are validated through a wind farm power maximization problem using a medium-fidelity model, demonstrating improved early transient performance compared to pure SFO.
	\end{abstract}
	
	\section{Introduction}
	Feedback optimization (FO) has emerged as a powerful framework for driving complex dynamical systems to optimal steady-state operation without requiring perfect model knowledge \cite{hauswirth2024optimization}. By replacing model-based steady-state predictions with real-time output measurements, feedback optimization algorithms can achieve optimality despite model inaccuracies and unknown disturbances. This approach has been extensively studied theoretically and applied in various applications including power system control \cite{dominguez2023online, menta2018stability}, process optimization \cite{francois2013measurement} and robotics control \cite{carnevale2024nonconvex}.

A central challenge in implementing feedback optimization is obtaining the sensitivity matrix that relates steady-state output to control inputs. For linear systems, this sensitivity is directly available from the system matrices. However, for nonlinear systems with high-dimensional state spaces, such as the flow models governing wind farm aerodynamics \cite{boersma2018control}, computing this sensitivity exactly could be computationally prohibitive. This difficulty has motivated the development of approximation strategies.

Existing approximation approaches can be classified into three categories. Model-based methods compute the sensitivity by linearizing the system around a fixed operating point \cite{colombino2019towards,ortmann2020experimental}, but the approximation quality depends critically on the choice of linearization point. Sequential feedback optimization (SFO) methods address this by adaptively updating the linearization point during optimization \cite{huang2025sequential}, eliminating the need for a-priori operating point selection. However, during the early transient, when the measured state can be far from the steady state of the current input, the linearized sensitivity may not provide effective optimization directions. Model-free methods estimate gradient information purely from measurements. These include recursive least squares (RLS) approaches that learn the sensitivity online from incremental input-output data \cite{picallo2022adaptive}, and extremum seeking control (ESC) methods that estimate the objective gradient directly using periodic perturbations and demodulation \cite{ariyur2003real}. While such estimates are formed from the realized plant response and therefore do not inherit any model approximation error, they typically require long learning periods to obtain accurate gradient information. Recent hybrid approaches attempt to improve performance by combining different types of gradient estimates. \cite{he2024online} combines the RLS approach with a zeroth-order FO method that directly estimates the objective gradient \cite{he2023model}, though zeroth-order methods may suffer from high-variance estimates that can lead to slow convergence or instability in practice.

This paper proposes a hybrid extension of SFO, termed H-SFO, that augments the model-based SFO gradient with a correction term. We study two variants, H-SFO-RLS and H-SFO-ESC, based on RLS sensitivity estimation and ESC gradient estimation, respectively. The key idea is to combine the model-based and correction gradients through a convex combination with weights that decay to zero and are summable. This strategy allows the algorithm to benefit from the correction during the early transient while gradually shifting toward the model-based linearization, so that the accuracy of the correction matters less as the weights decay.

The main contributions are as follows: First, we establish that both H-SFO-RLS and H-SFO-ESC converge to a neighborhood of the optimal steady state, with the correction term acting as a vanishing perturbation. Second, numerical simulations on a wind farm control problem show that the proposed methods accelerate early power increase compared to SFO while maintaining comparable steady-state performance.

\section{Problem Formulation}\label{section:problem}

Consider a discrete-time nonlinear system
\begin{equation}\label{plant_model}
x_{k+1}=f(x_k,u_k),\qquad y_k=g(x_k,u_k),
\end{equation}
where $x_k\in\mathbb{R}^n$ is the state, $u_k\in\mathcal{U}\subseteq\mathbb{R}^p$ is the control input constrained to a convex compact set $\mathcal{U}$, and $y_k\in\mathbb{R}^m$ is the measured output. Both $f:\mathbb{R}^n\times\mathbb{R}^p\to\mathbb{R}^n$ and $g:\mathbb{R}^n\times\mathbb{R}^p\to\mathbb{R}^m$ are continuously differentiable.

When a constant control $\bar{u}$ is applied, the corresponding steady state $\bar{x}$ is defined implicitly by $\bar{x}=f(\bar{x},\bar{u})$. We denote the corresponding steady-state mapping by $\phi:\mathcal{U}\to\mathbb{R}^n$, i.e., $\bar{x} = \phi(\bar{u})$. This induces the steady-state input-output mapping $h(\bar{u}):=g(\phi(\bar{u}),\bar{u})$, which relates the control input $\bar{u}$ to steady-state output $\bar{y}$.

We are interested in the constrained steady-state optimization problem
\begin{equation}\label{eq:optimization}
\min_{\bar{u}\in\mathcal{U},\,\bar{y}\in\mathbb{R}^m}\ J(\bar{u},\bar{y})\quad\text{subject to}\quad \bar{y}=h(\bar{u}),
\end{equation}
where $J:\mathcal{U}\times\mathbb{R}^m\to\mathbb{R}$ is a cost function. This formulation captures various applications, including wind farm power maximization, where $x_k$ represents flow velocities, $u_k$ contains turbine control parameters (e.g., induction factors and yaw angles), and $y_k$ measures power production subject to wake interaction dynamics encoded in $f$ \cite{boersma2018control}.

We make the following assumptions, which are commonly used in the feedback optimization literature \cite{bianchin2023online, hauswirth2020timescale, he2023model} to ensure well-posedness and enable convergence analysis.
\begin{assum}\label{ass:system}
The system dynamics in \eqref{plant_model} and the steady-state mapping $h$ satisfy:
\begin{enumerate}
\item[(i)] The mapping $x \mapsto f(x,u)$ is uniformly contractive, i.e., there exists $\rho_f \in (0,1)$ such that $\|f(x_1,u) - f(x_2,u)\|\le \rho_f \|x_1 - x_2\|$ for all $x_1,x_2\in\mathbb{R}^n$ and $u \in \mathbb{R}^p$.

\item[(ii)] The partial gradients $\nabla_x f(x,u)$ and $\nabla_u f(x,u)$ are uniformly bounded (i.e., $\|\nabla_u f(x,u)\|\le G_u^f$ for some $G_u^f>0$) and Lipschitz continuous with constants $L_{f,x}$ and $L_{f,u}$, respectively.

\item[(iii)] The steady-state input-output mapping $h: \mathcal{U} \to \mathbb{R}^m$ is $L_h$-Lipschitz continuous.
\end{enumerate}
\end{assum}
Assumption~\ref{ass:system}(i) also guarantees that the steady-state mapping $\phi$ introduced above is well defined. Indeed, for every fixed $\bar u\in\mathcal U$, the mapping $x\mapsto f(x,\bar u)$ is a contraction, so the Banach fixed-point theorem guarantees a unique fixed point $\phi(\bar u)$. Moreover, $\|\nabla_x f\|\le\rho_f<1$ implies that $I-\nabla_x f$ is invertible. Applying the local implicit function theorem at each $(\phi(\bar u),\bar u)$ shows that $\phi$ is continuously differentiable.
\begin{assum}\label{ass:cost}
The cost function $J$ in~\eqref{eq:optimization} satisfies:
\begin{enumerate}
\item[(i)] There exist constants $G_u^J$, $G_y^J$, $L_{J,u}$, $L_{J,y}>0$ such that $\|\nabla_u J(u,y)\|\le G_u^J$, $\|\nabla_y J(u,y)\|\le G_y^J$, and the partial gradients are Lipschitz continuous in $(u,y)$ with constants $L_{J,u}$ and $L_{J,y}$.

\item[(ii)] For any fixed $y$, the map $u\mapsto\nabla_u J(u,y)$ is $\mu_J$-strongly monotone for some $\mu_J>0$.
\end{enumerate}
\end{assum}

Under these assumptions, problem \eqref{eq:optimization} admits an optimal solution $(\bar{u}^\ast,\bar{y}^\ast)$. 

Traditional feedforward optimization requires explicit knowledge of $h$, which is often unavailable or computationally prohibitive. Feedback optimization overcomes this limitation by using real-time output measurements $y_k$ instead of steady-state predictions. The ideal feedback optimization (FO) scheme assumes access to the exact steady-state sensitivity $\nabla h(u_k)$ and performs a projected gradient step: 
\begin{equation}\label{eq:idealFO}
\begin{aligned}
y_k &= g(x_k,u_k),\\
d_k &= \nabla_u J(u_k,y_k) + \nabla h(u_k)^\top \nabla_y J(u_k,y_k),\\
u_{k+1} &= \mathrm{proj}_{\mathcal{U}}\!\left\{u_k - \alpha\, d_k\right\},\\
x_{k+1} &= f(x_k,u_k),
\end{aligned}
\end{equation}
where $\alpha>0$ is the step size and $\mathrm{proj}_{\mathcal{U}}$ denotes the Euclidean projection onto $\mathcal{U}$. 

For nonlinear systems with high-dimensional state spaces, computing the sensitivity exactly can be challenging, since it requires computing the steady state $\phi(u_k)$, i.e., solving $\bar x = f(\bar x, u_k)$, at every iteration. This motivates the development of sequential feedback optimization~\cite{huang2025sequential}, whose key idea is to linearize the dynamics \eqref{plant_model} around the current operating point, available from measurements, and compute the sensitivity of the linearized system. This sequential linearization approach is computationally tractable and adapts to the current state. However, the linearization is taken at the measured state $x_k$ rather than at the steady state $\phi(u_k)$ associated with the current input, and the accuracy of the resulting sensitivity degrades with the mismatch between the two, which can be large during the early transient.

Alternatively, data-driven methods such as RLS \cite{picallo2022adaptive} and ESC \cite{choi2002extremum} can estimate gradient information from measurements and thus do not rely on such a linearization, but typically need substantial exploration and may converge slowly. The proposed hybrid scheme (H-SFO) combines these complementary approaches. The model-based SFO gradient is computed from a local linearization at the current state and becomes increasingly accurate as the trajectory settles. In contrast, the RLS and ESC corrections are computed from finite input variations, namely realized input increments for RLS and sinusoidal perturbations for ESC. They therefore provide an alternative source of gradient information during the early transient, when the linearization error can be large. By combining them through time-varying weights, H-SFO aims to exploit the advantages of both strategies.
\section{Hybrid SFO Design}\label{sec:algs}


This section presents the model-based and correction components of H-SFO and summarizes the complete algorithm. We denote by $(\hat{x}_k, \hat{u}_k, \hat{y}_k)$ the state, control, and output trajectories generated by the hybrid algorithm. The model-based component computes the gradient using sequential linearization around the current operating point $(\hat{x}_k, \hat{u}_k)$. Specifically, the linearized sensitivity at iteration $k$ is
\begin{align}
H_{\mathrm{lin},k} &:= \nabla_x g(\hat x_k,\hat u_k)\big(I-\nabla_x f(\hat x_k,\hat u_k)\big)^{-1}\nabla_u f(\hat x_k,\hat u_k)\notag\\
&\quad\   + \nabla_u g(\hat x_k,\hat u_k),\label{lin_sensitivity}
\end{align}
which approximates the steady-state sensitivity $\nabla h(\hat{u}_k)$. The corresponding SFO gradient is then computed by
\begin{equation}\label{eq:sfo_gradient}
d_{\mathrm{SFO},k} = \nabla_u J(\hat{u}_k,\hat{y}_k) + H_{\mathrm{lin},k}^\top \nabla_y J(\hat{u}_k,\hat{y}_k).
\end{equation}

The correction gradient provides alternative optimization directions. We consider two variants that differ in what information they estimate and how the correction gradient is constructed.

\textbf{H-SFO-RLS (Hybrid SFO with Recursive Least Squares).} This variant combines the model-based linearized sensitivity with an adaptive sensitivity estimate, learned online from incremental input-output data by the RLS algorithm of \cite{picallo2022adaptive}. The key observation is that the output increments $\Delta \hat{y}_k := \hat{y}_k - \hat{y}_{k-1}$ can be related to input increments $\Delta \hat{u}_k := \hat{u}_k - \hat{u}_{k-1}$ through the linear regression model $\Delta\hat{y}_k = U_{\Delta, k-1}\text{vec}(\nabla h(\hat{u}_{k-1})) + \omega_{m,k-1}$, with regressor $U_{\Delta, k-1}:= \Delta \hat{u}_{k-1}^\top\otimes I_m$ and bounded residual $\omega_{m,k-1}$. The estimate $H_{\mathrm{RLS},k}$ is obtained from this model by the recursion in \cite[eq. (9)]{picallo2022adaptive} with positive-definite tuning matrices $\Sigma_p$ and $\Sigma_m$, and the RLS-based gradient is computed as
\begin{equation}\label{eq:gradient_rls}
d_{\mathrm{RLS},k} = \nabla_u J(\hat{u}_k,\hat{y}_k) + H_{\mathrm{RLS},k}^\top \nabla_y J(\hat{u}_k,\hat{y}_k).
\end{equation}

\textbf{H-SFO-ESC (Hybrid SFO with Extremum Seeking Control).} This variant constructs an ESC-based correction without explicitly estimating the sensitivity matrix. With sinusoidal perturbations $\delta_k:= [a_1\sin(\omega_1 k),\dots,a_p\sin(\omega_p k)]^\top$ at distinct frequencies, we define
\begin{equation*}
J_k^{\delta}:=J(\hat u_k, g(f(\hat{x}_k,u_k^\delta), u_k^\delta)),\quad u_k^\delta:= \text{proj}_{\mathcal{U}}(\hat{u}_k + \delta_k).
\end{equation*}
Following the filtering and demodulation structure of ESC \cite[Sec.~II]{choi2002extremum}, let $s_k:= [\sin(\omega_1 k),\dots, \sin(\omega_p k)]^\top$ and define
\begin{equation}\label{eq:gradient_esc}
d_{\mathrm{ESC},k} = \nabla_u J(\hat{u}_k,\hat{y}_k) + L(z)[s_kH(z)[J_k^\delta]],
\end{equation}
where $H(z)$ and $L(z)$ denote stable high- and low-pass filters respectively. Thus, H-SFO-ESC directly constructs the correction $d_{\mathrm{ESC},k}$ without estimating $\nabla h(\hat{u}_k)$.

For both approaches, the correction gradient is combined with the SFO gradient in~\eqref{eq:sfo_gradient} through a convex combination with time-varying weight $\lambda_k\in [0,1]$ (see Algorithm~\ref{alg:HSFO}). The weight starts close to one so that the algorithm initially relies more on the correction term, which may provide more effective optimization directions when the linearization error is large. As $\lambda_k$ decays to zero, the algorithm gradually transitions to the model-based linearization, which becomes increasingly accurate as the trajectory approaches the optimal steady state. By choosing $(\lambda_k)_{k\in\mathbb{N}}$ to be summable, the correction term acts as a vanishing perturbation, which is key for the convergence analysis.

\begin{algorithm}
\caption{Hybrid Sequential Feedback Optimization (H-SFO)}\label{alg:HSFO}
\textbf{Given:} stepsize $\alpha>0$, weights $\{\lambda_k\}$, initial $(\hat x_0,\hat u_0)$.\\
\textbf{For $k=0,1,2,\dots$}
\begin{enumerate}
\item Measure output: $\hat y_k=g(\hat x_k,\hat u_k)$.
\item Compute linearized sensitivity $H_{\mathrm{lin},k}$ via \eqref{lin_sensitivity}.
\item Compute SFO gradient via \eqref{eq:sfo_gradient}.
\item Compute correction gradient $d_{\mathrm{corr},k}$:\\
\textbf{H-SFO-RLS}: Update sensitivity $H_{\mathrm{RLS},k}$ using the RLS recursion in \cite[eq.~(9)]{picallo2022adaptive}, and set $d_{\mathrm{corr},k}$ via \eqref{eq:gradient_rls}.\\
\textbf{H-SFO-ESC}: Update $d_{\mathrm{corr},k}$ via \eqref{eq:gradient_esc}.
\item Form combined gradient: 
\[\hat{d}_k = (1-\lambda_k)d_{\mathrm{SFO},k} + \lambda_k d_{\mathrm{corr},k}.\]
\item Update $\hat u_{k+1}$ via projected gradient:
\begin{equation}\label{eq:u_update}
\hat{u}_{k+1} = \mathrm{proj}_{\mathcal{U}}(\hat u_k - \alpha \hat d_k + \omega_{u,k}),
\end{equation}
where $\omega_{u,k}$ is an excitation signal ($\omega_{u,k} = 0$ for H-SFO-ESC).
\item Propagate the system state: $\hat x_{k+1}=f(\hat x_k,\hat u_k)$.
\end{enumerate}
\end{algorithm}

To guarantee identifiability for the RLS estimator, we impose the following excitation condition \cite[Sec. 3.2]{picallo2022adaptive}.
\begin{assum}\label{ass:excitation}
For H-SFO-RLS, the exploration signal $\omega_{u,k}$ is uniformly bounded, i.e., $\sup_k \|\omega_{u,k}\| \le \bar{\omega} < \infty$, and the resulting input increments are persistently exciting: there exists $T\in\mathbb{N}$ such that for all $t>0$, the matrix formed by columns $\Delta \hat{u}_{t+i}$ for $i\in\{0,\dots,T\}$ has full rank, i.e.,
\begin{equation}\label{eq:PE_design}
\mathrm{rank}(\Delta\hat{u}_t,\dots,\Delta\hat{u}_{t+T}) = p.
\end{equation}
\end{assum}
Following \cite{picallo2022adaptive}, condition \eqref{eq:PE_design} is imposed on the closed-loop increments, while the probing signal $\omega_{u,k}$ is introduced to provide the required excitation. Conditions for persistent excitation and the design of sufficiently rich probing signals are discussed in \cite{bai1985persistency,yuan1977probing}. 

We choose the combination weights $(\lambda_k)_{k\in\mathbb{N}}$ such that
\begin{equation}\label{eq:alpha_design}
\lambda_k\in[0,1],\qquad \lim_{k\to\infty}\lambda_k=0,\qquad \sum_{k=0}^{\infty}\lambda_k<\infty.
\end{equation}
This choice ensures that the correction term acts as a summable perturbation in the convergence analysis.

\emph{Relation to existing methods.} The combination weight $\lambda_k$ allows the hybrid scheme to interpolate between purely model-based and purely adaptive approaches. When $\lambda_k\equiv 0$, both H-SFO-RLS and H-SFO-ESC reduce to the SFO method. Conversely, when $\lambda_k\equiv 1$, H-SFO-RLS becomes a purely adaptive FO method similar to \cite{picallo2022adaptive} with a dynamic plant, and H-SFO-ESC becomes pure ESC.

The most closely related work in the literature is \cite{he2024online}, where the authors propose a hybrid FO method that combines RLS sensitivity estimates with a model-free gradient estimate, and focus on dynamic regret analysis. In contrast, our H-SFO combines either RLS or ESC with a model-based linearized gradient from sequential linearization, enabling a convergence analysis of the control sequence. 
\section{Convergence Analysis}\label{sec:conv}
This section establishes the convergence of Algorithm~\ref{alg:HSFO} by relating the H-SFO iterates to the ideal feedback optimization scheme in \eqref{eq:idealFO}. For the convergence analysis, we set $g(x,u)=x$ to simplify notation (the general case follows by incorporating the Lipschitz constants of $g$ and $\nabla g$), so that $h=\phi$. 

\subsection{Auxiliary lemmas}
We first recall two results from \cite[Lemmas~3.2 and 3.4]{huang2025sequential} that quantify the behavior of the model-based sensitivity, and of the ideal FO scheme.
\begin{lem}\label{lem:lipschitz}
Let Assumption~\ref{ass:system} be satisfied. Then the linearized sensitivity in \eqref{lin_sensitivity} is Lipschitz continuous over $\mathbb{R}^n\times \mathcal{U}$: 
\begin{align}\label{eq:lipHlin}
\|H_{\mathrm{lin}}(x_1,u_1)-H_{\mathrm{lin}}(x_2,u_2)\|&\le C_h^{\mathrm{lin}}(\|x_1-x_2\|\notag\\
&\quad+\|u_1-u_2\|),\notag
\end{align}
with $C_h^{\mathrm{lin}}:= \frac{(1-\rho_f)L_{f,u}+G_u^f L_{f,x}}{(1-\rho_f)^2}$.
\end{lem}

\noindent 	Note that $\nabla h(u) = H_{\mathrm{lin}}(\phi(u),u)$, Lemma~\ref{lem:lipschitz} gives
\begin{equation}\label{eq:linerr}
\|\nabla h(\hat{u}_k) - H_{\mathrm{lin},k}\|\le C_h^{\mathrm{lin}}\hat e_k,\quad \hat e_k := \|\hat x_k - \phi(\hat{u}_k)\|,
\end{equation}
which quantifies the sensitivity error discussed in Section~\ref{section:problem}.
\begin{lem}\label{lem:ideal}
Let Assumptions~\ref{ass:system}–\ref{ass:cost} hold and $\alpha>0$ be chosen so that the matrix
\[
M:=\begin{bmatrix}
\sqrt{1-2\alpha\mu_J+\alpha^2 L_{J,u}^2} + \alpha C_1 & \alpha C_2\\
G_u^f & \rho_f
\end{bmatrix},
\]
satisfies $\rho(M)<1$, where $C_1:=L_h L_{J,y}+G_y^J C_h^{\mathrm{lin}}(1+L_h)$ and $C_2:=L_{J,u}+L_h L_{J,y}$. Then the ideal FO scheme in \eqref{eq:idealFO} converges to $(\bar u^\ast,\bar y^\ast)$.
\end{lem}

Since $M$ is nonnegative $2\times 2$ matrix, condition $\rho(M)<1$ holds if and only if the leading principal minors of $I - M$ are positive. In particular, a sufficiently small $\alpha$ satisfies this condition whenever $(\mu_J - C_1)(1-\rho_f) > C_2G_u^f$.

Next, we establish the properties of the correction terms. The following lemma derives, for the plant \eqref{plant_model}, the incremental regression model that is assumed in \cite{picallo2022adaptive}.
\begin{lem}\label{lem:lpv_bounded_Sigma}
Let Assumption~\ref{ass:system} hold and define the steady-state sensitivity $H_k:=\nabla h(\hat{u}_k)$ and its vectorization $h_k:=\mathrm{vec}(H_k)$. For the increments
\[
\Delta \hat{x}_k:=\hat{x}_k-\hat{x}_{k-1},\quad \Delta \hat{u}_k:=\hat{u}_k-\hat{u}_{k-1},\quad \Delta \hat{y}_k:=\hat{y}_k-\hat{y}_{k-1},
\]
there exist matrices $A_{k-1},B_{k-1}$, with $\|A_{k-1}\|\le \rho_f$ and $\|B_{k-1}\|\le G_u^f$, such that 
\begin{align}
h_k &= h_{k-1} + \omega_{p,k-1}, \label{eq:reg_process}\\
\Delta \hat{y}_k \; &=\; U_{\Delta,k-1}\,h_{k-1} \;+\; \omega_{m,k-1}, \label{eq:reg_measurement}
\end{align}
with $U_{\Delta,k-1}:=\Delta \hat{u}_{k-1}^{\!\top}\!\otimes I_m$, $\omega_{p,k-1} := h_k - h_{k-1}$ and
\begin{align}
\omega_{m,k-1}\;&:=\;A_{k-1}\,\Delta \hat{x}_{k-1} \;+\; (B_{k-1}-H_{k-1})\,\Delta \hat{u}_{k-1}.\label{eq:omegam_def}
\end{align}
Moreover, $\sup_k\|\omega_{p,k}\| < \infty$ and $\sup_k\|\omega_{m,k}\| < \infty$.
\end{lem}
See Appendix~\ref{pf:LPV_structure} for the proof.

\begin{lem}\label{lem:corr_bounded}
Let Assumptions~\ref{ass:system}--\ref{ass:cost} hold.
\begin{enumerate}
\item[(i)] (H-SFO-RLS) Let Assumption~\ref{ass:excitation} additionally hold, and let $H_{\mathrm{RLS},k}$ be generated by \cite[eq.~(9)]{picallo2022adaptive} applied to \eqref{eq:reg_process}--\eqref{eq:reg_measurement}, with tuning matrices satisfying $\underline\sigma I\preceq\Sigma_p,\Sigma_m\preceq\bar\sigma I$ for some $\bar\sigma\ge\underline\sigma>0$. Then there exists $C_{\mathrm{RLS}}>0$ such that
$\sup_{k\ge 0}\|H_{\mathrm{RLS},k}-\nabla h(\hat u_k)\|\le C_{\mathrm{RLS}}$.
\item[(ii)] (H-SFO-ESC) Let $H(z)$ and $L(z)$ in \eqref{eq:gradient_esc} be stable. There exists $C_{\mathrm{ESC}}>0$ such that $\sup_{k\ge 0}\|d_{\mathrm{ESC},k}\|\le C_{\mathrm{ESC}}$.
\end{enumerate}
\end{lem}
See Appendix~\ref{pf:uniform_bound} for the proof.
Part (i) follows by the argument in \cite[proof of Proposition~1]{picallo2022adaptive}: condition \eqref{eq:PE_design} and $\Sigma_p\succ0$ make the recursion uniformly completely observable and controllable, implying exponentially stable homogeneous error dynamics \cite[Ch.~7]{jazwinski2007stochastic}. The bounded residuals of Lemma~\ref{lem:lpv_bounded_Sigma} then imply a uniformly bounded estimation error. Part (ii) follows since Assumption~\ref{ass:system}(i) and the compactness of $\mathcal{U}$ imply that $(\hat x_k)_k$ is bounded, hence $J_k^\delta$ is bounded by continuity of $f$, $g$ and $J$. Together with the bounded cost gradient of Assumption~\ref{ass:cost}(i), the bounded sinusoidal signal $s_k$, and the stability of $H(z)$ and $L(z)$, this implies that $d_{\mathrm{ESC},k}$ is uniformly bounded.

\subsection{Main convergence result for H-SFO}
We now combine the model-based and correction components to characterize the behavior of H-SFO. Compared with the SFO analysis in \cite{huang2025sequential}, the convergence proof must explicitly account for the additional error introduced by the correction term and the probing signal.
\begin{thm}\label{thm:hyb}
Let Assumptions~\ref{ass:system}–\ref{ass:cost} hold, and let the combination weights $(\lambda_k)_{k\in\mathbb{N}}$ satisfy \eqref{eq:alpha_design}. Choose $\alpha>0$ in \eqref{eq:u_update} so that the spectral condition in Lemma~\ref{lem:ideal} is satisfied. For
H-SFO-RLS, let Assumption \ref{ass:excitation} additionally hold, and for H-SFO-ESC, set $\bar\omega = 0$. Then the control input sequence $(\hat{u}_k)_{k\in\mathbb{N}}$ generated by the H-SFO algorithm (either H-SFO-RLS or H-SFO-ESC) satisfies
\begin{equation*}
\limsup_{k\to\infty} \|\hat u_k - \bar u^\ast\|\ \le\ \frac{\alpha\, G_y^JC_h^{\mathrm{lin}}\, L_h (\alpha \bar G + \bar\omega) + (1-\rho_f)\bar\omega}{(1-\rho_f)(1-\rho(M))},
\end{equation*}
where $\bar G:=G_u^J+G_y^JG_u^f/(1-\rho_f)\ge \|d_{\mathrm{SFO},k}\|$. 
\end{thm}

See Appendix~\ref{pf:convergence} for the proof. Theorem~\ref{thm:hyb} shows that H-SFO converges to a neighborhood of the optimal steady state. The correction term acts as a summable perturbation, which may influence the transient behavior but does not contribute an additional term to the asymptotic bound. For H-SFO-RLS, however, the persistent probing signal contributes the terms depending on $\bar\omega$ above, which vanish for H-SFO-ESC since $\bar\omega=0$. Since the bound is derived using worst-case estimates for the tracking error and the deviation from the ideal FO trajectory, it may be conservative, especially when $\rho_f$ or $\rho(M)$ is close to one.	
	
\section{Numerical Results}\label{sec:simulation}

This section evaluates the proposed H-SFO methods through numerical simulations on a wind farm power maximization problem. We first compare both H-SFO-RLS and H-SFO-ESC with the baseline SFO algorithm to assess the transient performance, then study the influence of the combination weight decay rate, and finally compare them with an alternative hybrid feedback optimization method.

\subsection{Simulation setup}

All simulations use a wind farm consisting of $9$ NREL 5MW reference turbines arranged in a $3\times 3$ square grid with inter-turbine spacing. The farm dynamics are simulated using WFSim \cite{boersma2018control}, a medium-fidelity dynamic flow model based on two-dimensional Navier-Stokes equations (see \cite{boersma2018control} for detailed model description). The computational domain spans $2518.8\,\mathrm{m} \times 1558.4\,\mathrm{m}$ and is discretized on a $50\times 25$ staggered grid. Steady uniform flow conditions are assumed with $u_b = 8\,\mathrm{m/s}$ and $v_b = 0\,\mathrm{m/s}$.

The control inputs for turbine $i$ are the disk-based thrust coefficient $C_{T,i}'$ and yaw angle $\gamma_i$. Let $y = h(C_T', \gamma) = [P_1,\ldots, P_N]^{\top}$ denote the steady-state input-output mapping of WFSim, where $P_i$ represents the power generated by turbine $i$. Since WFSim is nonlinear and high-dimensional, it is not practical to compute an explicit expression for this mapping. Following \cite{huang2025sequential}, the steady-state power maximization problem is formulated as
\begin{equation}\label{eq:wf_problem}
\begin{aligned}
\min_{C_T', \gamma} \  & J(C_T', \gamma, y) = \left(\frac{\mathbf{1}_N^{\top} y - P^{\mathrm{ref}}}{P^{\mathrm{ref}}}\right)^2 + \tfrac{\mu}{2}\|C_T'\|^2 + \tfrac{\mu_\gamma}{2}\|\gamma\|^2 \\
\text{s.t.} \  & y = h(C_T', \gamma),\ \ C_T' \in [0.4, 3.6]^N,\ \ \gamma \in [-30^{\circ}, 30^{\circ}]^N, 
\end{aligned}
\end{equation}
where we set $P^{\mathrm{ref}} = 18\,\mathrm{MW}$, and choose regularization parameters $\mu = 2.8\times 10^{-4}$ and $\mu_\gamma = 2\times 10^{-5}$. All controllers are initialized from the greedy (non-cooperative) steady state with $C_{T,i}' = 2$ and $\gamma_i = 0^{\circ}$ for all turbines.

The model-based sensitivity is computed at each iteration by linearizing the WFSim dynamics around the current operating point \cite{huang2025sequential}, and the RLS algorithm estimates $\nabla_{C_T'} h$ and $\nabla_\gamma h$ online. In H-SFO-RLS, $\omega_{u,k}$ is a truncated Gaussian probing signal with standard deviation $10^{-3}$ for $C_T'$ and $10^{-2}$ degrees for $\gamma$. In H-SFO-ESC, the dither amplitudes are $a_{C_T'}=0.05$ and $a_\gamma=2^\circ$, with frequencies in $[0.005,0.03]$ rad/s and filter bandwidths $0.002$ rad/s.

\subsection{Comparison with SFO}

Fig.~\ref{fig:hfo_vs_sfo} compares the total farm power trajectories for the proposed H-SFO-RLS, H-SFO-ESC, the SFO method, and the greedy baseline. Both hybrid methods and SFO achieve substantially higher power than the greedy controller, with hybrid methods attaining a moderately higher steady-state power than SFO. Indeed, Theorem~\ref{thm:hyb} only guarantees convergence to a neighborhood of the optimum, so different trajectories may settle at different operating points. However, their transient behaviors differ significantly. SFO increases power slowly in the early phase and then exhibits a large overshoot before settling to its steady state. This may be related to the fact that, during the initial transient, the measured flow field is far from the steady state of the current input, so that $H_{\mathrm{lin},k}$ may be a poor approximation of $\nabla h$, as shown in \eqref{eq:linerr}.

In contrast, both H-SFO-RLS and H-SFO-ESC show improved transient performance. A possible explanation is that the exploration underlying the RLS and ESC corrections provides more effective optimization directions during the early phase, when the operating point is far from the cooperative optimum. By combining these corrections with the model-based SFO gradient through the weight $\lambda_k$, the algorithm benefits from both the structural information of the linearized model and the adaptive corrections. This results in a faster overall power increase and smoother convergence.

\begin{figure}[htbp]
\centering
\includegraphics[width=\linewidth]{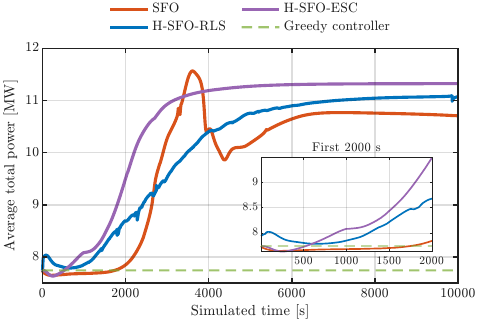}
\caption{Total farm power comparison: H-SFO versus baseline SFO and greedy controller.}
\label{fig:hfo_vs_sfo}
\end{figure}

\subsection{Influence of weight decay rate}

 Fig.~\ref{fig:weight} investigates the influence of the combination weights in~\eqref{eq:alpha_design} on H-SFO-RLS (top row) and H-SFO-ESC (bottom row). We study two weight schedules: the asymptotic decay $\lambda_k = 1/(1 + (k/200)^p)$ with exponent $p \in \{1.5, 2, 2.5\}$, and the finite-time switching $\lambda_k = \max\{1-k/T,0\}^2$ with transition time $T\in\{1000, 2000, 3000\}$. 

For H-SFO-RLS, the algorithm achieves immediate power improvement from the initial iterations, maintaining performance above the greedy baseline throughout. Both schedules exhibit trade-offs between convergence speed and smoothness. With asymptotic decay (left panel), slower decay ($p=1.5$) maintains a high weight on the RLS correction longer, leading to faster initial power gain but also more oscillations in the transient. Faster decay ($p=2.5$) reduces RLS influence more quickly, resulting in a behavior closer to pure SFO with less transient improvement. The intermediate choice $p=2$ offers a reasonable balance, achieving a fast transient while maintaining a smooth convergence. With finite-time switching (right panel), longer transition time ($T=3000$) yields the fastest initial power increase and smooth convergence, while shorter transition time ($T=1000$) exhibits slower initial response. The intermediate choice $T=2000$ shows behavior between these extremes.

For H-SFO-ESC, the power increases more smoothly. With the asymptotic decay schedule (left panel), faster decay ($p=2.5$) leads to quicker convergence, while slower decay ($p=1.5$) maintains ESC correction longer, yielding slightly faster initial gains but extending the overall convergence period. With the finite-time switching schedule (right panel), shorter transition time ($T=1000$) results in faster convergence to steady state. Unlike H-SFO-RLS, all H-SFO-ESC trajectories remain smooth, since H-SFO-ESC does not include the additive probing signal.

Both weight schedules successfully improve early transient performance compared to pure SFO for both methods, though the specific transient characteristics differ. The weight parameters should be tuned according to the specific application requirements.

\begin{figure}[htbp]
\centering
\includegraphics[width=\linewidth]{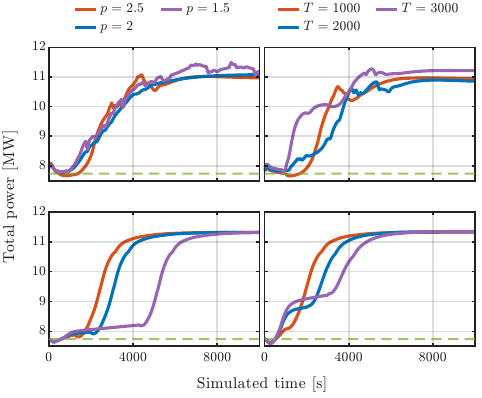}
\caption{Influence of the combination weights on H-SFO-RLS (top) and H-SFO-ESC (bottom): asymptotic decay with exponent $p$ (left) and finite-time switching with transition time $T$ (right).}
\label{fig:weight}
\end{figure}

\subsection{Comparison with alternative hybrid FO method}

We also compare H-SFO with the hybrid framework proposed in \cite{he2024online}, which combines the RLS estimate with a model-free, zeroth-order gradient estimate. For this comparison, we present trajectories of the control inputs rather than total power output because the zeroth-order gradient estimates induce extreme power fluctuations in the fixed step-size case, making power-based visualization uninformative. Figure~\ref{fig:he_case} shows the control trajectories of upstream turbines for the method of \cite{he2024online} under two step-size strategies. 

\begin{figure}[htbp]
\centering
\includegraphics[width=\linewidth]{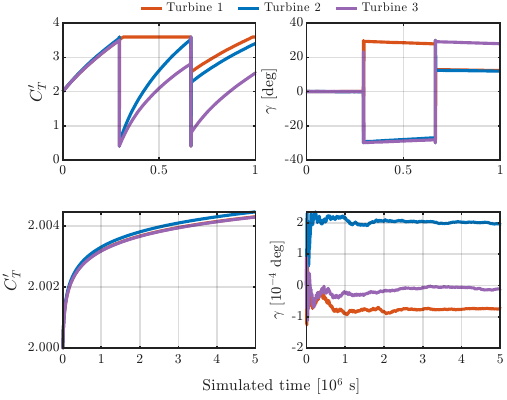}
\caption{Control trajectories of upstream turbines for the zeroth-order gray-box method of \cite{he2024online} with fixed step size (top) and diminishing step size (bottom).}
\label{fig:he_case}
\end{figure}

The top row gives results with small fixed step-sizes $\alpha_{\mathrm{C_T'}} = 10^{-4}$ and $\alpha_\gamma = 10^{-3}$. Both the disk-based thrust coefficients $C_T'$ and yaw angles $\gamma$ exhibit severe fluctuations and frequently hit constraint boundaries. This instability occurs because the model-free method produces high-variance gradient estimates when facing the complex nonlinear wake dynamics.  The bottom row uses a diminishing step-size $\alpha_k \propto 1/k$ to stabilize the controller. In this case, the trajectories are more regular, but convergence becomes extremely slow. Even after $5\times 10^6$ iterations, the control inputs remain close to their initial greedy values, yielding almost the same power production as the greedy controller.

These results demonstrate that, for the considered wind farm control problem in~\eqref{eq:wf_problem}, this hybrid FO method with fixed step-sizes leads to instability, while with diminishing step-sizes it converges too slowly to be practical. In contrast, the model-based approach in H-SFO appears to provide more reliable optimization directions, enabling both stability and faster convergence.

\section{Conclusion}\label{section:conclusion}

We introduced a hybrid extension of sequential feedback optimization (H-SFO) that augments the model-based scheme with a vanishing correction term while retaining convergence to a neighborhood of the optimal steady state. The wind farm case study demonstrates that H-SFO has the potential to accelerate the increase in total power production, while the combination weights provide a simple way to balance the correction against convergence smoothness.

Future work will focus on validating the approach in high-fidelity models (such as SOWFA) and wind tunnel testing. An extension to dynamic models with time-varying parameters is also of interest.

\useRomanappendicesfalse
\appendices
\section{Preliminary bounds}\label{pf:prelim}
Throughout the Appendix we use the setting of Section~\ref{sec:conv}, i.e., $g(x,u)=x$, so that $\hat y_k=\hat x_k$ and $h=\phi$. Moreover, $\|\cdot\|$ denotes the Euclidean norm for vectors and the induced $2$-norm for matrices, and $\|\cdot\|_F$ denotes the Frobenius norm. The following facts are used in the proofs.

First, uniform contractivity implies $\|\nabla_x f(x,u)\|\le\rho_f$ for all $(x,u)$. Hence, by the Neumann series, $\|(I-\nabla_xf)^{-1}\|\le 1/(1-\rho_f)$ and
\begin{equation}\label{eq:Hlin_bound}
\|H_{\mathrm{lin}}(x,u)\|\le \frac{G_u^f}{1-\rho_f},\qquad \forall (x,u)\in\mathbb{R}^n\times\mathcal U .
\end{equation}
Second, since $\nabla h(u)=H_{\mathrm{lin}}(\phi(u),u)$ and $\phi$ is $L_h$-Lipschitz, Lemma~\ref{lem:lipschitz} gives
\begin{equation}\label{eq:gradh_lip}
\|\nabla h(u_1)-\nabla h(u_2)\|\le C_h^{\mathrm{lin}}(1+L_h)\|u_1-u_2\| ,
\end{equation}
for all $u_1,u_2\in\mathcal U$, and $\|\nabla h(u)\|\le L_h$. Third, let $D_{\mathcal U}:=\max_{u,u'\in\mathcal U}\|u-u'\|$ denote the diameter of $\mathcal U$, which is finite since $\mathcal U$ is compact. Then $\|\Delta\hat u_k\|\le D_{\mathcal U}$ for all $k$.

\begin{lem}\label{lem:aux_bounds}
Let Assumption~\ref{ass:system} hold. Then the steady-state error defined in \eqref{eq:linerr} satisfies 
\begin{align}
\hat e_{k+1}&\le \rho_f\,\hat e_k+L_h\|\hat u_{k+1}-\hat u_k\| ,\label{eq:track_rec}\\
\hat e_k&\le \rho_f^{k}\hat e_0+\frac{L_hD_{\mathcal U}}{1-\rho_f},\label{eq:track_crude}\\
\|\Delta\hat x_{k}\|&\le \rho_f^{k-1}\|\Delta\hat x_{1}\|+\frac{G_u^fD_{\mathcal U}}{1-\rho_f},\quad k\ge1 .\label{eq:dx_bound}
\end{align}
In particular, $(\hat x_k)_k$ is bounded and $\|\Delta\hat x_k\|\le\bar\Delta_x:=\|\Delta\hat x_1\|+G_u^fD_{\mathcal U}/(1-\rho_f)$ for all $k\ge1$.
\end{lem}

\begin{proof}
Since $\phi(\hat u_k)=f(\phi(\hat u_k),\hat u_k)$, Assumption~\ref{ass:system}(i) and (iii) give
\begin{align*}
\hat e_{k+1}&=\|f(\hat x_k,\hat u_k)-\phi(\hat u_{k+1})\|\\
&\le \|f(\hat x_k,\hat u_k)-f(\phi(\hat u_k),\hat u_k)\|+\|\phi(\hat u_k)-\phi(\hat u_{k+1})\| ,
\end{align*}
which is \eqref{eq:track_rec}. Then \eqref{eq:track_crude} follows by unrolling and using $\|\hat u_{k+1}-\hat u_k\|\le D_{\mathcal U}$. For \eqref{eq:dx_bound}, Assumption~\ref{ass:system}(i)--(ii) give
\begin{align*}
\|\Delta\hat x_k\|&=\|f(\hat x_{k-1},\hat u_{k-1})-f(\hat x_{k-2},\hat u_{k-2})\|\\
&\le\rho_f\|\Delta\hat x_{k-1}\|+G_u^f\|\Delta\hat u_{k-1}\| ,
\end{align*}
and \eqref{eq:dx_bound} follows by unrolling. Finally, $(\hat x_k)_k$ is bounded by \eqref{eq:track_crude}, since $\phi(\mathcal U)$ is compact.
\end{proof}

The following elementary result is used to handle the vanishing weights.

\begin{lem}\label{lem:conv}
Let $\rho\in(0,1)$ and let $(a_t)_{t\ge0}$ be nonnegative.
\begin{enumerate}
\item[(a)] If $\sum_{t\ge0}a_t<\infty$, then $\sum_{t=0}^k\rho^{k-t}a_t\to0$ as $k\to\infty$.
\item[(b)] If there exists a constant $\bar a$ such that $\limsup_{t\to\infty}a_t\le\bar a$, then $\limsup_{k\to\infty}\sum_{t=0}^k\rho^{k-t}a_t\le \bar a/(1-\rho)$.
\end{enumerate}
\end{lem}

\begin{proof}
(a) Fix $\varepsilon>0$ and choose $N$ with $\sum_{t>N}a_t<\varepsilon$. For $k>N$,
\[
\sum_{t=0}^k\rho^{k-t}a_t\le \rho^{k-N}\sum_{t=0}^{N}a_t+\sum_{t>N}a_t<\rho^{k-N}\sum_{t=0}^{N}a_t+\varepsilon ,
\]
and the first term tends to zero as $k\to\infty$. Since $\varepsilon$ is arbitrary, the claim follows. (b) Fix $\varepsilon>0$ and choose $N$ with $a_t\le\bar a+\varepsilon$ for all $t>N$. Then, for $k>N$, $\sum_{t=0}^k\rho^{k-t}a_t\le\rho^{k-N}\sum_{t=0}^{N}a_t+(\bar a+\varepsilon)/(1-\rho)$, and letting $k\to\infty$ and then $\varepsilon\to0$ gives the claim.
\end{proof}

\section{Proof of Lemma~\ref{lem:lpv_bounded_Sigma}}\label{pf:LPV_structure}
For $k\ge 2$, the plant model in \eqref{plant_model} gives
\[
\Delta \hat{x}_k = f(\hat{x}_{k-1},\hat{u}_{k-1})-f(\hat{x}_{k-2},\hat{u}_{k-2}).
\]
By applying the integral mean-value formula along $z(\tau):=(\hat{x}_{k-2}+\tau \Delta \hat{x}_{k-1},\,\hat{u}_{k-2}+\tau \Delta \hat{u}_{k-1})$ with $\tau\in [0,1]$, we obtain
\begin{equation}\label{eq:dx_mvt}
\Delta \hat{x}_k = A_{k-1}\Delta \hat{x}_{k-1} + B_{k-1}\Delta \hat{u}_{k-1},
\end{equation}
with
\[
A_{k-1}=\int_0^1 \nabla_x f(z(\tau))\,d\tau,\qquad
B_{k-1}=\int_0^1 \nabla_u f(z(\tau))\,d\tau ,
\]
so that $\|A_{k-1}\|\le\rho_f$ and $\|B_{k-1}\|\le G_u^f$ by Assumption~\ref{ass:system}. Adding and subtracting $H_{k-1}\Delta \hat{u}_{k-1}$ gives
\[
\Delta \hat{x}_k = H_{k-1}\Delta \hat u_{k-1} + \omega_{m,k-1},
\]
with $\omega_{m,k-1}$ defined in \eqref{eq:omegam_def}. Since $\Delta \hat y_k=\Delta \hat{x}_k$ and $H_{k-1}\Delta \hat{u}_{k-1} = U_{\Delta,k-1}\,h_{k-1}$, \eqref{eq:reg_measurement} holds. By the definition of $\omega_{p,k-1}$, the evolution of $h_k$ satisfies \eqref{eq:reg_process}.

It remains to bound the two residuals. Using $\|H_{k-1}\|=\|\nabla h(\hat u_{k-1})\|\le L_h$, $\|\Delta\hat u_{k-1}\|\le D_{\mathcal U}$ and the bound $\|\Delta\hat x_{k-1}\|\le\bar\Delta_x$ of Lemma~\ref{lem:aux_bounds},
\begin{equation}\label{eq:wm_bound}
\|\omega_{m,k-1}\|\le \rho_f\bar\Delta_x+(G_u^f+L_h)D_{\mathcal U}=:\bar\omega_m .
\end{equation}
Moreover, using $\|\mathrm{vec}(X)\|=\|X\|_F\le\sqrt{p}\,\|X\|$ for $X\in\mathbb{R}^{m\times p}$ together with \eqref{eq:gradh_lip},
\begin{align}
\|\omega_{p,k-1}\|&=\|\nabla h(\hat u_k)-\nabla h(\hat u_{k-1})\|_F\notag\\
&\le \sqrt{p}\,C_h^{\mathrm{lin}}(1+L_h)D_{\mathcal U}=:\bar\omega_p .\label{eq:wp_bound}
\end{align}
Both bounds are finite and independent of $k$, which completes the proof. \hfill $\blacksquare$

\section{Proof of Lemma~\ref{lem:corr_bounded}}\label{pf:uniform_bound}
\emph{Part (i).} Write $U_k:=U_{\Delta,k}$ and $\hat h_k:=\mathrm{vec}(H_{\mathrm{RLS},k})$. The recursion of \cite[eq.~(9)]{picallo2022adaptive} applied to \eqref{eq:reg_process}--\eqref{eq:reg_measurement} reads
\begin{align}
K_k&=\Sigma_{k-1}U_{k-1}^\top\big(\Sigma_m+U_{k-1}\Sigma_{k-1}U_{k-1}^\top\big)^{-1},\notag\\
\hat h_k&=\hat h_{k-1}+K_k\big(\Delta\hat y_k-U_{k-1}\hat h_{k-1}\big),\label{eq:rls_rec}\\
\Sigma_k&=\big(I-K_kU_{k-1}\big)\Sigma_{k-1}+\Sigma_p,\qquad \Sigma_0\succ0,\notag
\end{align}
where $\Sigma_k\succ0$ is the matrix updated by the estimator recursion, and the tuning matrices satisfy $\underline\sigma I\preceq\Sigma_p\preceq\bar\sigma I$ and $\underline\sigma I\preceq\Sigma_m\preceq\bar\sigma I$. Defining the vectorized estimation error $\tilde e_k:=h_k-\hat h_k$ and substituting \eqref{eq:reg_process}--\eqref{eq:reg_measurement} into \eqref{eq:rls_rec},
\begin{equation}\label{eq:err_dyn}
\tilde e_k=(I-K_kU_{k-1})\tilde e_{k-1}+\omega_{p,k-1}-K_k\,\omega_{m,k-1}.
\end{equation}

Since $\hat u_k\in\mathcal U$ for all $k$ and $\mathcal U$ is compact, the regressors are uniformly bounded, $\|U_{k-1}\|=\|\Delta\hat u_{k-1}\|\le D_{\mathcal U}$. Moreover, following \cite[proof of Proposition~1]{picallo2022adaptive}, the persistence-of-excitation condition \eqref{eq:PE_design} implies that, for a sufficiently large window $T$, there exist $\beta_1,\beta_2>0$ such that
\[
\beta_1 I\ \preceq\ \sum_{i=k}^{k+T}U_{\Delta,i}^\top U_{\Delta,i}\ \preceq\ \beta_2 I,\qquad\forall k\ge0 .
\]
Defining the windowed information matrix $W_I(k,T):=\sum_{i=k}^{k+T}U_{\Delta,i}^\top\Sigma_m^{-1}U_{\Delta,i}$ and using $\underline\sigma I\preceq\Sigma_m\preceq\bar\sigma I$,
\[
\frac{\beta_1}{\bar\sigma}I\ \preceq\ W_I(k,T)\ \preceq\ \frac{\beta_2}{\underline\sigma}I,\qquad\forall k\ge0 ,
\]
so that the regression model \eqref{eq:reg_process}--\eqref{eq:reg_measurement} is uniformly completely observable \cite[p.~35]{sastry2011adaptive}, while $\Sigma_p\succeq\underline\sigma I$ makes it uniformly completely controllable. Consequently, the homogeneous part of \eqref{eq:err_dyn} is exponentially stable \cite[Ch.~7]{jazwinski2007stochastic}, i.e., its state transition matrices satisfy
\begin{equation}\label{eq:Phi_bound}
\|\Phi(k,t)\|\le D\rho_R^{\,k-t},\qquad \forall k\ge t\ge 0,
\end{equation}
for some $D\ge1$ and $\rho_R\in(0,1)$. The same bounds also imply that $(\Sigma_k)_k$ is uniformly bounded, i.e., $\Sigma_k\preceq\sigma_{\max}I$ for some $\sigma_{\max}>0$, and therefore
\[
\|K_k\|\le \frac{\sigma_{\max}D_{\mathcal U}}{\underline\sigma}=:\bar K,\qquad\forall k\ge0 .
\]

By unrolling \eqref{eq:err_dyn} we obtain
\[
\tilde e_k=\Phi(k,0)\tilde e_0+\sum_{t=1}^{k}\Phi(k,t)\big(\omega_{p,t-1}-K_t\,\omega_{m,t-1}\big),
\]
and hence, using \eqref{eq:Phi_bound}, \eqref{eq:wm_bound}, \eqref{eq:wp_bound} and the geometric-series bound $\sum_{t=1}^k\rho_R^{\,k-t}\le 1/(1-\rho_R)$,
\begin{equation}\label{eq:etilde_bound}
\|\tilde e_k\|\ \le\ \underbrace{D\rho_R^{\,k}\|\tilde e_0\|}_{\text{initial error}}+\frac{D(\bar\omega_p+\bar K\bar\omega_m)}{1-\rho_R}\ \le\ \tilde C ,
\end{equation}
with $\tilde C:=D\|\tilde e_0\|+D(\bar\omega_p+\bar K\bar\omega_m)/(1-\rho_R)$. Moreover, since $\|H_{\mathrm{RLS},k}-\nabla h(\hat u_k)\|\le\|H_{\mathrm{RLS},k}-\nabla h(\hat u_k)\|_F=\|\tilde e_k\|$, the bound \eqref{eq:etilde_bound} gives $\sup_{k\ge0}\|H_{\mathrm{RLS},k}-\nabla h(\hat u_k)\|\le C_{\mathrm{RLS}}:=\tilde C$. 

\emph{Part (ii).} By Lemma~\ref{lem:aux_bounds}, $(\hat x_k)_k$ is bounded, and $u_k^\delta=\mathrm{proj}_{\mathcal U}(\hat u_k+\delta_k)\in\mathcal U$. Since $f$, $g$ and $J$ are continuous and $\mathcal U$ is compact, we obtain $|J_k^\delta|\le\bar J$ for some $\bar J<\infty$. The filters $H(z)$ and $L(z)$ are design choices of the ESC scheme, and we choose them BIBO stable, so that their impulse responses are absolutely summable. Let $\|H\|_{\ell_1}$ and $\|L\|_{\ell_1}$ denote the (finite) $\ell_1$-norms of the impulse responses of the stable filters $H(z)$ and $L(z)$, and let $\eta_H,\eta_L$ bound their zero-input responses. Since $\|s_k\|_\infty\le1$, each component of $L(z)[s_kH(z)[J_k^\delta]]$ is bounded by $\|L\|_{\ell_1}(\|H\|_{\ell_1}\bar J+\eta_H)+\eta_L$. Together with $\|\nabla_uJ\|\le G_u^J$ from Assumption~\ref{ass:cost}(i), \eqref{eq:gradient_esc} gives
\[
\|d_{\mathrm{ESC},k}\|\le G_u^J+\sqrt{p}\,\big(\|L\|_{\ell_1}(\|H\|_{\ell_1}\bar J+\eta_H)+\eta_L\big)=:C_{\mathrm{ESC}} ,
\]
for all $k\ge0$. Note that the extremum-tracking property of \cite{choi2002extremum} is not used here, since it requires additional conditions on the dither amplitude and on the time-scale separation. \hfill $\blacksquare$

\section{Proof of Theorem~\ref{thm:hyb}}\label{pf:convergence}
Let $(u_k,y_k)$ denote the ideal FO trajectory generated by \eqref{eq:idealFO} and $(\hat u_k,\hat y_k)$ the H-SFO trajectory, and let
\[
d(u,y):= \nabla_u J(u,y) + \nabla h(u)^{\top}\nabla_y J(u,y),
\]
so that $d_k=d(u_k,y_k)$.

\emph{Step 1 (comparison of the input sequences).} By the non-expansiveness of $\mathrm{proj}_{\mathcal{U}}$, the ideal FO update in \eqref{eq:idealFO} and the H-SFO update in \eqref{eq:u_update},
\begin{align}
\|u_{k+1}-\hat u_{k+1}\|
&\le \big\|(u_k-\hat u_k) - \alpha(d_k-\hat d_k)-\omega_{u,k}\big\|\notag\\
&\le \big\|(u_k-\hat u_k) - \alpha\big(d_k-d(\hat u_k,\hat y_k)\big)\big\|\notag\\
&\quad + \alpha\,\underbrace{\big\|d(\hat u_k,\hat y_k)-\hat d_k\big\|}_{\Delta_k}+\bar\omega ,
\label{eq:triangle_split}
\end{align}
where $\bar\omega=0$ for H-SFO-ESC. To bound the first term of \eqref{eq:triangle_split}, we split it as
\begin{align}
&\big\|(u_k-\hat u_k) - \alpha\big(d(u_k,y_k)-d(\hat u_k,\hat y_k)\big)\big\|\notag\\
&\le \big\|(u_k-\hat u_k)-\alpha\big(\nabla_uJ(u_k,y_k)-\nabla_uJ(\hat u_k,y_k)\big)\big\|\notag\\
&\quad+\alpha\big\|\nabla_uJ(\hat u_k,y_k)-\nabla_uJ(\hat u_k,\hat y_k)\big\|\notag\\
&\quad+\alpha\big\|\nabla h(u_k)^\top\nabla_yJ(u_k,y_k)-\nabla h(\hat u_k)^\top\nabla_yJ(\hat u_k,\hat y_k)\big\| .\label{eq:step1split}
\end{align}
By Assumption~\ref{ass:cost}, the first term of \eqref{eq:step1split} is bounded by $\sqrt{1-2\alpha\mu_J+\alpha^2L_{J,u}^2}\,\|u_k-\hat u_k\|$, using $\mu_J$-strong monotonicity and the $L_{J,u}$-Lipschitz continuity of $\nabla_uJ$, and the second one by $\alpha L_{J,u}\|y_k-\hat y_k\|$. For the third term, adding and subtracting $\nabla h(u_k)^\top\nabla_yJ(\hat u_k,\hat y_k)$ and using $\|\nabla h\|\le L_h$, \eqref{eq:gradh_lip} and Assumption~\ref{ass:cost}(i), we obtain the bound $\alpha\big(L_hL_{J,y}+G_y^JC_h^{\mathrm{lin}}(1+L_h)\big)\|u_k-\hat u_k\|+\alpha L_hL_{J,y}\|y_k-\hat y_k\|$. Collecting the three terms, with $C_1$ and $C_2$ as in Lemma~\ref{lem:ideal}, we get
\begin{equation}\label{eq:uk_ineq_Delta}
\|u_{k+1}-\hat u_{k+1}\|\le M_{11}\|u_k - \hat u_k\| + \alpha C_2\|y_k - \hat y_k\| + \alpha \Delta_k+\bar\omega ,
\end{equation}
where $M_{11}$ denotes the $(1,1)$ entry of the matrix $M$ in Lemma~\ref{lem:ideal}.

\emph{Step 2 (bound on $\Delta_k$).} Recall that the hybrid gradient is
\begin{equation}\label{eq:hatg_split}
\hat d_k = (1-\lambda_k)\,d_{\mathrm{SFO},k} + \lambda_k\,d_{\mathrm{corr},k},
\end{equation}
with $d_{\mathrm{corr},k}=d_{\mathrm{RLS},k}$ for H-SFO-RLS and $d_{\mathrm{corr},k}=d_{\mathrm{ESC},k}$ for H-SFO-ESC. By substituting \eqref{eq:hatg_split} into the definition of $\Delta_k$ and using $\lambda_k\in[0,1]$,
\begin{align}
\Delta_k
&\le (1-\lambda_k)\underbrace{\big\|d(\hat u_k,\hat y_k) - d_{\mathrm{SFO},k}\big\|}_{\Delta_k^{\mathrm{SFO}}}\notag\\
&\quad+ \lambda_k\underbrace{\big\|d(\hat u_k,\hat y_k) - d_{\mathrm{corr},k}\big\|}_{\Delta_k^{\mathrm{corr}}}\notag\\
&\le \Delta_k^{\mathrm{SFO}}+\lambda_k\Delta_k^{\mathrm{corr}} .\label{eq:Delta_split_var}
\end{align}
From \eqref{eq:sfo_gradient}, Assumption~\ref{ass:cost}(i) and \eqref{eq:linerr},
\begin{align}
\Delta_k^{\mathrm{SFO}}
&= \big\|\big(\nabla h(\hat u_k) - H_{\mathrm{lin},k}\big)^{\top}\nabla_y J(\hat u_k,\hat y_k)\big\|\notag\\
&\le G_y^J\big\|\nabla h(\hat u_k) - H_{\mathrm{lin},k}\big\|\le G_y^JC_h^{\mathrm{lin}}\hat e_k . \label{eq:Delta_SFO_step1}
\end{align}
For H-SFO-RLS, by the definition of $d_{\mathrm{RLS},k}$ in \eqref{eq:gradient_rls} and Lemma~\ref{lem:corr_bounded}(i),
\begin{align}
\Delta_k^{\mathrm{corr}}
&= \big\|(\nabla h(\hat u_k)-H_{\mathrm{RLS},k})^{\top}\nabla_y J(\hat u_k,\hat y_k)\big\|\notag\\
&\le G_y^J\|\nabla h(\hat u_k)-H_{\mathrm{RLS},k}\|\le G_y^JC_{\mathrm{RLS}} .\label{eq:Dcorr_RLS_1}
\end{align}
For H-SFO-ESC, by Assumptions~\ref{ass:system}(iii) and \ref{ass:cost}(i), $\sup_{k\ge0}\|d(\hat u_k,\hat y_k)\|\le G_u^J+L_hG_y^J$, so that Lemma~\ref{lem:corr_bounded}(ii) gives
\begin{equation}\label{eq:Dcorr_ESC_final}
\Delta_k^{\mathrm{corr}}\le \|d(\hat u_k,\hat y_k)\|+\|d_{\mathrm{ESC},k}\|\le G_u^J + L_hG_y^J + C_{\mathrm{ESC}} .
\end{equation}
In both cases,
\begin{equation}\label{eq:Delta_unified}
\Delta_k \le G_y^J C_h^{\mathrm{lin}} \hat e_k + \lambda_k C_{\Delta},
\end{equation}
where $C_\Delta$ denotes the corresponding constant in \eqref{eq:Dcorr_RLS_1} or \eqref{eq:Dcorr_ESC_final}, i.e., $C_\Delta:=G_y^JC_{\mathrm{RLS}}$ for H-SFO-RLS and $C_\Delta:=G_u^J+L_hG_y^J+C_{\mathrm{ESC}}$ for H-SFO-ESC.

\emph{Step 3 (tracking mismatch of the hybrid trajectory).} By \eqref{eq:u_update}, the non-expansiveness of $\mathrm{proj}_{\mathcal U}$ and $\hat u_k=\mathrm{proj}_{\mathcal U}(\hat u_k)$, we have $\|\hat u_{k+1}-\hat u_k\|\le\alpha\|\hat d_k\|+\bar\omega$. By \eqref{eq:Hlin_bound} and Assumption~\ref{ass:cost}(i), $\|d_{\mathrm{SFO},k}\|\le \bar G$, while Lemma~\ref{lem:corr_bounded} gives $\|d_{\mathrm{corr},k}\|\le C_G$, with $C_G:=G_u^J+G_y^J(L_h+C_{\mathrm{RLS}})$ for H-SFO-RLS and $C_G:=C_{\mathrm{ESC}}$ for H-SFO-ESC. Hence, by \eqref{eq:hatg_split},
\[
\|\hat u_{k+1}-\hat u_k\|\le\alpha\big(\bar G+\lambda_k C_G\big)+\bar\omega ,
\]
and \eqref{eq:track_rec} yields, after unrolling,
\[
\hat e_{k}\le\rho_f^{k}\hat e_0+L_h(\alpha \bar G+\bar\omega)\sum_{t=0}^{k-1}\rho_f^{\,k-1-t}+\alpha L_hC_G\sum_{t=0}^{k-1}\rho_f^{\,k-1-t}\lambda_t .
\]
The first term vanishes, the second is bounded by $L_h(\alpha\bar G+\bar\omega)/(1-\rho_f)$, and the third tends to zero by Lemma~\ref{lem:conv}(a), since $\sum_t\lambda_t<\infty$ by \eqref{eq:alpha_design}. Therefore,
\begin{equation}\label{eq:ebar}
\limsup_{k\to\infty}\hat e_k\ \le\ \bar e:=\frac{L_h(\alpha \bar G+\bar\omega)}{1-\rho_f},
\end{equation}
and $(\hat e_k)_k$ is bounded.

\emph{Step 4 (comparison system).} The output-error dynamics of the plant satisfy
\begin{equation}\label{F:y}
\|y_{k+1}-\hat y_{k+1}\|\le G_u^f\|u_k-\hat u_k\|+\rho_f\|y_k-\hat y_k\| .
\end{equation}
By stacking $\epsilon_k:=[\|u_k-\hat u_k\|,\ \|y_k-\hat y_k\|]^\top$ and combining \eqref{eq:uk_ineq_Delta}, \eqref{eq:Delta_unified} and \eqref{F:y}, we obtain the componentwise inequality
\begin{equation}\label{F:stack}
\epsilon_{k+1}\le M \epsilon_k + \begin{bmatrix}\alpha G_y^JC_h^{\mathrm{lin}} \hat e_k + \alpha\lambda_kC_\Delta+\bar\omega\\0\end{bmatrix},
\end{equation}
with $M$ defined in Lemma~\ref{lem:ideal}. Since $M$ is nonnegative and, by $\alpha C_2>0$ and $G_u^f>0$, irreducible, the Perron--Frobenius theorem implies the existence of a positive vector $w = [w_1,w_2]^{\top}$ such that $M^\top w= \rho(M) w$. Multiplying \eqref{F:stack} by $w^\top$, which preserves the inequality since $w>0$, and defining $V_k := w^{\top}\epsilon_k$, we have
\[
V_{k+1} \le \rho(M) V_k + w_1\big(\alpha G_y^JC_h^{\mathrm{lin}}\hat e_k + \alpha\lambda_k C_\Delta+\bar\omega\big).
\]
By unrolling this recursion,
\begin{align}
V_{k+1} &\le (\rho(M))^{k+1} V_0 + w_1\alpha G_y^JC_h^{\mathrm{lin}}\sum_{t=0}^k (\rho(M))^{k-t} \hat e_t\notag\\
&\quad + w_1\alpha C_\Delta\sum_{t=0}^k (\rho(M))^{k-t} \lambda_t+w_1\bar\omega\sum_{t=0}^k(\rho(M))^{k-t} .\label{eq:unroll}
\end{align}

\emph{Step 5 (limit).} We now apply Lemma~\ref{lem:conv} with $\rho=\rho(M)$ to the three sums in \eqref{eq:unroll}. By \eqref{eq:ebar} and Lemma~\ref{lem:conv}(b), the first sum is bounded by $\bar e/(1-\rho(M))$ in the limit. By \eqref{eq:alpha_design} and Lemma~\ref{lem:conv}(a), the second sum tends to zero. The third sum is bounded by $1/(1-\rho(M))$, and $(\rho(M))^{k+1}V_0\to0$. It follows that
\[
\limsup_{k\to\infty} V_k \le \frac{w_1\big(\alpha G_y^JC_h^{\mathrm{lin}} \bar e+\bar\omega\big)}{1-\rho(M)} .
\]
Since $w_1\|u_k - \hat{u}_k\|\le V_k$, we obtain
\begin{align}
\limsup_{k\to\infty} \|u_k-\hat u_k\|
&\le \frac{\alpha\, G_y^JC_h^{\mathrm{lin}}\bar e+\bar\omega}{1-\rho(M)}\notag\\
&=\frac{\alpha\, G_y^JC_h^{\mathrm{lin}}\, L_h (\alpha \bar G + \bar\omega) + (1-\rho_f)\bar\omega}{(1-\rho_f)(1-\rho(M))}.
\end{align}
Combining with the convergence of the ideal FO scheme to $(\bar u^\ast,\bar y^\ast)$ in Lemma~\ref{lem:ideal} concludes the proof. \hfill $\blacksquare$

\bibliographystyle{IEEEtran}
\bibliography{cdcconf.bib}	
	
\end{document}